\documentclass[12pt]{article}
\usepackage{amsmath,amsfonts,amssymb}
\usepackage{alltt}

\usepackage{hyperref}

\usepackage{amsmath,amsfonts,ifthen,fullpage,enumerate,float}  

\usepackage{amsmath,amsfonts,amssymb}
\usepackage{graphicx}        
\def\trace{\mathop{\rm trace}\nolimits}

\def\bmat#1{\begin{bmatrix}#1\end{bmatrix}}
\def\pmat#1{\begin{pmatrix}#1\end{pmatrix}}

\def\question#1{{\bf Question: }#1}
\def\question#1{}

\def\C{{\operatorname{C}}}

\def\SL{\mathop{\it SL}\nolimits}

\def\Z{\mathbb{Z}}
\def\CC{\mathbb{C}}

\def\ZZ{\mathbb{Z}}

\def\Cd{\CC^d}

\def\Zd{\Z_d}

\def\CC{\mathbb{C}}

\def\implies{\Longrightarrow}

\newcommand{\RR}{\mathbb{R}}

\newtheorem{theorem}{Theorem}[section]
\newtheorem{corollary}{Corollary}[section]
\newtheorem{lemma}{Lemma}[section]
\newtheorem{example}{Example}[section]

\newtheorem{proposition}{Proposition}[section]
\newtheorem{definition}{Definition}[section]

\newenvironment{proof}{{\noindent \it
Proof.}}{\hfill$\Box$\medskip}
\newif\ifdraft\def\draft{\drafttrue\hoffset=.8truecm\showlabeltrue
\def\comment##1{{\bf comment: ##1}}
\headline={\sevenrm \hfill \ifx\filenamed\undefined\jobname\else\filenamed\fi%
(.tex) (as of \ifx\updated\undefined???\else\updated\fi)
 \TeX'ed at {\hour\time\divide\hour by 60{}%
\minutes\hour\multiply\minutes by 60{}%
\advance\time by -\minutes
\the\hour:\ifnum\time<10{}0\fi\the\time\  on \today\hfill}}
}

\def\inpro#1{\langle#1\rangle}
\def\ip#1{\langle\kern-.28em\langle#1\rangle\kern-.28em\rangle_\nu}

\def\cU{{\cal U}}

\def\norm#1{\Vert#1\Vert}
\def\openR{{{\rm I}\kern-.16em {\rm R}}}

\let\ga\alpha

\let\gd\delta

\let\gth\theta

\let\gl\lambda

\let\gL\Lambda

\let\go\omega
\let\gO\Omega
\let\ga\alpha

\let\gd\delta

\def\inpro#1{\langle#1\rangle}

\def\Im{\mathop{\rm Im}\nolimits}

\def\Iff{\hskip1em\Longleftrightarrow\hskip1em}

\def\makeblanksquare#1#2{
\dimen0=#1pt\advance\dimen0 by -#2pt
      \vrule height#1pt width#2pt depth0pt\kern-#2pt
      \vrule height#1pt width#1pt depth-\dimen0 \kern-#1pt
      \vrule height#2pt width#1pt depth0pt \kern-#2pt
      \vrule height#1pt width#2pt depth0pt
}

\title{\bf Equations for the overlaps of a SIC}

\date{\today}

\author{ Len Bos$^*$ and Shayne Waldron$^\dagger$
 \\   \\
\vbox{
\hbox{\footnotesize ${}^*$ Department of Computer Science, University of Verona, Italy}
{\vskip 0.10 truecm}
\hbox{\footnotesize \noindent $^\dagger$ Department of Mathematics, University of Auckland,
Private Bag 92019, Auckland, New Zealand}
} }

\begin{document}

\maketitle 

\begin{abstract}

We give a holomorphic quartic polynomial
in the overlap variables whose zeros on the torus are precisely the 
Weyl-Heisenberg SICs (symmetric informationally complete positive
operator valued measures). 
By way of comparison, all the other known systems of equations
that determine a Weyl-Heisenberg SIC involve variables and their complex
conjugates.
We also give a related interesting result about the powers of the projective 
Fourier transform of the group $G=\Zd\times\Zd$.
\end{abstract}

\bigskip
\vfill


 \ifdraft	Algebraic geometry over groups; equations over groups \else\fi
81P15, \ifdraft	Quantum measurement theory \else\fi
81Q10, \ifdraft	Selfadjoint operator theory in quantum theory, including spectral analysis \else\fi
81R05, \ifdraft	Finite-dimensional groups and algebras motivated by physics and their representations \else\fi
\quad

secondary
05B30, \ifdraft (Other designs, configurations) \else\fi
42C15, \ifdraft (Series of general orthogonal functions, generalised Fourier expansions, non-orthogonal expansions)\else\fi
51F25. \ifdraft	Orthogonal and unitary groups \else\fi

\vskip .5 truecm
\hrule
\newpage

`
\section{Introduction}

Throughout fix the integer $d\ge2$, and let
$\go$ be the primitive $d$-th root of unity $\go^:=e^{{2\pi\over d}i}$.
We think of vectors in $\Cd$ as periodic signals on the group
$\Zd$, and hence index vectors and matrices by elements of $\Zd$.
A set of $d^2$ unit vectors $(v_j)$ in $\Cd$ (or the lines that they
determine) is said to be {\bf equiangular} if 
\begin{equation}
\label{equicdn}
|\inpro{v_j,v_k}|^2 = {1\over d+1}, \qquad j\ne k.
\end{equation}
In quantum information theory, the 
corresponding rank one orthogonal projections $(v_jv_j^*)$
are said to be a {\bf symmetric informationally complete positive operator 
valued measure}, or a {\bf SIC} for short.
The existence of a
SIC for every dimension $d$ is known as {\it Zauner's conjecture}
(from his 1999 thesis, see \cite{Zau10}), 
or as the {\it SIC problem}.

There are high precision numerical constructions of SICs
\cite{RBSC04}, \cite{SG09}, \cite{S17},
and exact SICs in various dimensions \cite{ACFW17}, \cite{GS17}.
In all of these constructions, the SIC is a {\bf Weyl-Heisenberg SIC},
i.e., is the orbit $(\rho(g)v)_{g\in G}$ of
a {\it fiducial vector} $v$ under the unitary irreducible projective 
representation 
$\rho:G\to\cU(\CC^{\Zd})$ of 
$G=\Zd\times\Zd$ 
with Schur multiplier $\ga$ given by 
\begin{equation}
\label{rhoalphadef}
\rho_{jk} = \rho((j,k))=S^j\gO^k, \qquad
\ga((j_1,j_2),(k_1,k_2))=\go^{j_2k_1},
\end{equation}
where $S$ is the cyclic shift matrix $S_{jk}:=\gd_{j,k+1}$ 
and $\gO$ is the diagonal
(modulation) matrix $\gO_{jk}:=\go^j\gd_{jk}$.
In this case, the equiangularity condition (\ref{equicdn}) becomes
\begin{equation}
\label{vcdns}
|\inpro{S^j\Omega^k v,v}|^2 = {1\over{d+1}}, \qquad 
(j,k)\ne(0,0).
\end{equation}

In this paper, we consider equations in the variables
\begin{equation}
\label{cjkdef}
c_{jk} 
=\inpro{S^j\gO^kv,v}=\trace(vv^* S^j\gO^k), \qquad
(j,k)\in\Zd\times\Zd,
\end{equation}
which determine a (Weyl-Heisenberg) SIC. These variables (or scalar multiples
of them) are called the {\bf overlaps} of the SIC. They
depend only on the fiducial projector $P=vv^*$.
The original attempts to find numerical and exact SIC fiducials 
(using Groebner basis methods)
involved polynomial equations
in the variables $v_0,\ldots,v_{d-1}$ and 
$\overline{v_0},\ldots,\overline{v_{d-1}}$, such as 
the equiangularity condition (\ref{vcdns}),
the equations (see \cite{BW07}, \cite{K08}, \cite{ADF07})
\begin{align}
\label{SICsimpeqns}
\sum_{r\in\Zd} v_r\overline{v}_{r+s}\overline{v}_{r+t}v_{r+s+t}
=\begin{cases}
0, & s,t\ne 0;\cr
{1\over d+1},  & s\ne0,t=0,\quad s=0,t\ne0;\cr
{2\over d+1}, & (s,t)=(0,0),
\end{cases}
\end{align}
and the variational characterisation (used for finding numerical SICs)
\begin{equation}
\label{vSICeqns2}
 {1\over d^2}\sum_{(j,k)\in\Zd^2} |\inpro{S^j\gO^k v,v}|^4
= {2\over d(d+1)} \norm{v}^4, \qquad \norm{v}^2=1. 
\end{equation}

More recent exact constructions of SICs 
\cite{ACFW17}
have been in the overlap variables $c_{jk}$
(utilising a natural Galois action on them).
Clearly the $c_{jk}$ giving a SIC fiducial projector $vv^*$ via (\ref{cjkdef})
must satisfy
\begin{equation}
\label{cjkcdns1}
c_{00}=\norm{v}^2=1, \qquad 
|c_{jk}|^2=|\inpro{S^j\Omega^k v,v}|^2={1\over d+1},
\quad (j,k)\ne(0,0), 
\end{equation}
and also, by the rule $\gO^kS^j=\go^{jk}S^j\gO^k$,
\begin{equation}
\label{cjkcdns2}
c_{jk} 
= \overline{\inpro{v,S^j\gO^kv}}
= \overline{\inpro{\gO^{-k}S^{-j}v,v}}
= \overline{\inpro{\go^{jk}S^{-j}\gO^{-k}v,v}}
= \go^{-jk} \overline{c_{-j,-k}}.
\end{equation}
These conditions on the overlap variables $c_{jk}$ are not enough to guarantee 
that they come from a fiducial projector $vv^*$ (and hence prove Zauner's conjecture).

In Section 2, we define a linear operator $T$, which is 
an example of 
the projective Fourier transform, which allows us
to reconstruct the fiducial projector as $vv^*=Tc$ from a suitable $c=(c_{jk})$.
We prove that
in addition to
(\ref{cjkcdns1}) and (\ref{cjkcdns2}),
the simple condition
$$\trace ((Tc)^4)=1 $$
ensures that a $c$ gives a SIC fiducial (Theorem \ref{cvcharthm}).
We then give some examples, and describe the action of the Clifford group
on the SIC fiducials give by overlaps $c$.

In Section 3, we give some interesting properties of $T$, i.e.,
the projective Fourier transform of $G=\Zd\times\Zd$. In particular,
we show that $(\sqrt{d}T)^{6d}=(-1)^{{1\over 2}d(d-1)}I$,
and a variant has order $4d$. 
To our knowledge, this is only 
the second example of a Fourier transform of finite order,
after
the (discrete) Fourier
transform for a finite abelian group $G=\Zd$ (which satisfies $F^4=I$).

In Section 4, we give another system of equations in the overlaps $c$
that determine a SIC.
These involve the symbol ($z$--transform) of the rows of $c$. 
The symbols for $c$ giving a SIC turn out
to have interesting Riesz-type factorisation properties.
We use these to describe the (sporadic) SICs for $d=3$, 
which are parametrised by a hypocycloid.

\section{The reconstruction operator}

Since the $\rho$ of (\ref{rhoalphadef}) is a unitary irreducible projective 
representation of dimension $d$, it follows that $(\rho(g))_{g\in G}$
is a tight frame
(called a nice error frame with index group $G$ \cite{CW17})
 for the $d\times d$ matrices
with the Frobenius inner product
$$ \inpro{A,B} := \trace(AB^*) = \sum_{j,k} a_{jk}\overline{b_{jk}}, $$
i.e.,
\begin{equation}
\label{tightframerec}
A={d\over|G|}\sum_{g\in G} \inpro{A,\rho(g)}\rho(g), \qquad
\forall A\in\CC^{d\times d}. 
\end{equation}
In this particular case,  
$(\rho(g))_{g\in G}=(S^j\gO^k)$ is an orthogonal basis.
Taking $A=vv^*$ above gives the following formula for reconstruction from the overlaps
$c_{jk}=\inpro{S^j\gO^kv,v}$
$$ vv^* 
= {d\over d^2}\sum_{j,k}\inpro{vv^*,S^{-j}\gO^{-k}} S^{-j}\gO^{-k}
= {1\over d}\sum_{j,k} \go^{jk}c_{jk} S^{-j}\gO^{-k}
= {1\over d}\sum_{j,k} c_{jk} (S^{j}\gO^{k})^*, $$
since $\go^{jk}S^{-j}\gO^{-k}=(S^j\gO^k)^*$,
and $(S^{-j}\gO^{-k})^* 
=\go^{jk}S^j\gO^k$ gives
$$\inpro{vv^*,S^{-j}\gO^{-k}}
=\trace(vv^*(S^{-j}\gO^{-k})^*)
=\trace(v^*\go^{jk}S^j\gO^kv)
=\go^{jk}\inpro{S^j\gO^kv,v}
= \go^{jk} c_{jk}. $$
Motivated by this, we define a linear map 
$T:\CC^{\Zd\times\Zd}\to\CC^{\Zd\times\Zd}$ by
\begin{equation}
\label{Tcdefn}
Tc := {1\over d}\sum_{j,k} c_{jk} (S^{j}\gO^{k})^*
= {1\over d}\sum_{j,k} \go^{jk} c_{jk} S^{-j}\gO^{-k}.
\end{equation}
This can 
be viewed as the $\ga$-Fourier transform of \cite{W18}
(for a Schur multiplier $\ga$) which is a map 
$F_\ga:\CC^G\to\oplus_\rho \CC^{d_\rho\times d_\rho}$,
where $\rho$ counts over the irreducible projective representations
of $G$ with multiplier $\ga$ (and dimension $d_\rho$). 
Here $G=\Zd\times\Zd$ has just one
such representation, the $\rho$ of (\ref{rhoalphadef}),
and $F_\ga$ of
$\nu=c\in\CC^G=\CC^{\Zd\times\Zd}$ at the unitary representation $\rho$ is
$$ (F_\ga\nu)_\rho = \sum_{g\in G}\nu(g)\rho(g)^*
=\sum_{j,k} c_{jk}(S^j\gO^k)^*
= d(Tc). $$
Thus $T$ is the projective Fourier transform for the group $G=\Zd\times\Zd$.
For this particular group, we can view the image of a vector in $\CC^G$
as being in $\CC^G=\CC^{\Zd\times\Zd}$, and as a result it is natural
to consider powers of the Fourier transform. The only other case that
we know of where this can be done is for the ordinary representations of
a finite abelian group (where the representations give the character
group $\hat G$, which can be identified with $G$). In this case the
(discrete) Fourier transform has order $4$.

We now use $T$ to characterise those vectors (matrices) 
$c\in\CC^G=\CC^{\Zd\times\Zd}$ which give a fiducial
projector $vv^*=Tc$.

\begin{lemma}
\label{evaluelemma}
Let $T$ be given by {\rm (\ref{Tcdefn})}.
Suppose that $c=(c_{jk})\in\CC^{\Zd\times\Zd}$ satisfies
\begin{enumerate}[\rm (i)]
\item $c_{jk}=\omega^{-jk}\overline{c_{-j,-k}}$
\item $c_{00}=1$
\item $|c_{jk}|^2={1\over d+1}$, $(j,k)\ne(0,0)$.
\end{enumerate}
Then $Tc$ is Hermitian, and its eigenvalues $\gl_1,\ldots,\gl_d$ satisfy
$$ \sum_j\gl_j=1, \qquad
 \sum_j\gl_j^2=1, \qquad
 \sum_{j\ne k}\gl_j\gl_k=0, $$
i.e., its characteristic polynomial has the form
$$ p_{Tc}(\gl) = \gl^d-\gl^{d-1} + 0\gl^{d-2}+a_{d-3}\gl^{d-3}+\cdots
+a_1\gl + a_0. $$
\end{lemma}

\begin{proof} Firstly, observe (i) implies that $Tc$ is Hermitian, since
$$ (Tc)^* = {1\over d}\sum_{j,k}\omega^{jk}c_{-j,-k} S^j\gO^k
= {1\over d}\sum_{j,k}\omega^{(-j)(-k)}c_{j,k} S^{-j}\gO^{-k}
= Tc. $$
Since $\trace(S^j\gO^k)=0$, $(j,k)\ne(0,0)$,
we calculate using (ii) that
$$  \sum_j\gl_j = \trace(Tc) = {1\over d}\sum_{j,k} \omega^{jk}c_{j,k} 
\trace(S^{-j}\gO^{-k})
= {1\over d}\omega^0 c_{00}\trace(I) =c_{00}=1. $$

The so called $2$--trace $\sum_{j\ne k}\gl_j\gl_k$ of $Tc$ is equal to
$\{(\trace(Tc))^2-\trace((Tc)^2)\}/2$.
Since $Tc$ is Hermitian, $\trace((Tc)^2)=\inpro{Tc,(Tc)^*}=\inpro{Tc,Tc}$,
and by the orthogonality of the $\rho_{jk}=S^j\gO^k$, we calculate
$$ \inpro{Tc,Tc} = {1\over d^2} \sum_{j,k} |c_{j,k}|^2 
\inpro{\rho_{jk}^*,\rho_{jk}^*}
= {1\over d}  \sum_{j,k} |c_{jk}|^2. $$
Now by (ii) and (iii),
$$ \trace((Tc)^2) = \inpro{Tc,Tc} = {1\over d}  \sum_{j,k} |c_{jk}|^2
= {1\over d}\Bigl( 1 + (d^2-1){1\over d+1}\Bigr) = 1, $$
and so
$\sum_{j\ne k}\gl_j\gl_k = \{(\trace(Tc))^2-\trace((Tc)^2)\}/2=(1-1)/2 =0$.
\end{proof}

\begin{theorem}
\label{cvcharthm}
Let $T$ be given by {\rm (\ref{Tcdefn})}.
Then a matrix $c=(c_{jk})\in\CC^{\Zd\times\Zd}$ determines
a fiducial projector for a Weyl-Heisenberg SIC by
$vv^*=Tc$ if and only if
\begin{enumerate}[\rm(i)]
\item $\displaystyle{ c_{jk}=\omega^{-jk}\overline{c_{-j,-k}} }$
\item $\displaystyle{ c_{00}=1 }$
\item $\displaystyle{|c_{jk}|^2={1\over d+1}}$, $(j,k)\ne(0,0)$
\item $\trace((Tc)^4)=1$
\end{enumerate}
Moreover this fiducial satisfies $\inpro{S^j\Omega^kv,v}=c_{jk},$ 
and if $v_0\neq0,$ then $v$ is given by
$$v={1\over d \overline{v_0}}\overline{c}\pmat{1\cr 1\cr \vdots\cr1}.$$
\end{theorem}

\begin{proof}
By Lemma \ref{evaluelemma},
the eigenvalues of the Hermitian matrix $Tc$ satisfy
$$ \sum_j\gl_j=1, \qquad \sum_j\gl_j^2=1, \qquad \sum_{j\ne k}\gl_j\gl_k=0, $$
so that $0\le\gl_j^2\le 1$. Thus $\gl_j^4\le\gl_j^2$, with equality if
and only if $\gl_j^2\in\{0,1\}$. But
$$ \sum_j \gl_j^4 = \trace((Tc)^4) =1 = \trace((Tc)^2)= \sum_j \gl_j^2, $$
so that $\gl_j^2\in\{0,1\}$, $\forall j$, and we must have $\gl_j=1$ for
some $j$, and $\gl_j=0$ for all others, i.e.,
$Tc$ is rank one, say 
$$ Tc = {1\over d} \sum_{j,k} c_{jk} \rho_{jk}^* = vv^*, \qquad
v\in\CC^{\Zd}. $$
Since 
$\{\rho_{jk}\}$ is orthogonal, taking the inner product
of the above 
with $\rho_{jk}=S^j\gO^k$ gives 
$$ c_{jk}={1\over d} c_{jk} \inpro{\rho_{jk}^*,\rho_{jk}^*} 
= \inpro{vv^*,\rho_{jk}^*}
= \trace(vv^*S^j\gO^) = \trace(v^* S^j\gO^k v) 
= \inpro{S^j\gO^kv,v}. $$
Finally, with $e_j$ the standard basis vectors, we calculate
\begin{align*}
\hbox{$j$--th entry of }\ \overline{c}\pmat{1\cr 1\cr \vdots\cr1}
&= \sum_k \overline{c_{jk}} 
= \sum_k \inpro{v,S^j\gO^kv}
= \inpro{v,S^j\Bigl( \sum_k \gO^k\Bigr)v} \cr
&= \inpro{v,S^j
\pmat{d&0&\cdots&0\cr0&0&\cdots&0\cr\vdots&\vdots&&\vdots\cr0&0&\cdots&0} v}
= \inpro{v,S^j dv_0 e_0}
= d\overline{v_0} \inpro{v,e_j}
= d\overline{v_0} v_j. 
\end{align*}
\end{proof}

\noindent
From the proof, we see that {\rm (iv)} can
be replaced by various 
equivalent conditions, e.g.,
\smallskip

{\rm (iv)${}'$}
{\sl The characteristic polynomial of $Tc$ has the form $P_{Tc}(\gl)=\gl^d-\gl^{d-1}$ }
\smallskip

{\rm (iv)$''$}
{\sl $\trace((Tc)^j)=1$, $j=1,2,\ldots$ }
\smallskip

\noindent
since given {\rm (i), (ii), (ii)},  
\smallskip

{\rm (iv)$''$ $\implies$ (iv) $\implies$
$Tc$ has eigenvalues $1,0,\ldots,0$ $\iff$ (iv)$'$ $\implies$ (iv)$''$.}
\smallskip

\noindent
By condition (ii), we may set $c_{00}=1$, to obtain the following
characterisation.

\begin{corollary}
The overlaps of a Weyl-Heisenberg SIC are precisely the
zeros of the polynomial $\trace((Tc)^4)=1$ on the torus 
$$|c_{jk}|={1\over\sqrt{d+1}}, \qquad
(j,k)\ne(0,0). $$
\end{corollary}

The condition (i) allows further variables $c_{jk}$ to be 
eliminated. When $d$ is odd, half of the  $(d^2-1)$ 
variables $c_{jk}$, $(j,k)\ne(0,0)$, can be eliminated. For $d$ even,
half of the $d^2-4$ variables $c_{jk}$, $(j,k)\not\in\{0,{d\over2}\}^2$,
can be eliminated, and
\begin{equation}
\label{deveneqs}
c_{{d\over2},0} 
= \overline{c_{{d\over2},0}},
\qquad 
c_{0,{d\over2}} 
= \overline{c_{0,{d\over2}}}, \qquad
c_{{d\over2},{d\over2}} 
= (-1)^{d\over2}\overline{c_{{d\over2},{d\over2}}}, 
\end{equation}
so that $c_{0,{d\over2}},c_{0,{d\over2}}\in\RR$, and
$c_{{d\over2},{d\over2}}$ is in $\RR$ for ${d\over2}$ even, 
and is in $i\RR$ for ${d\over2}$ odd. 

\begin{example}
For $d=2$, the conditions {\rm (\ref{deveneqs})} of {\rm (i)} give
$c_{01}\in\RR$, $c_{10}\in\RR$, $c_{11}\in i\RR$.
Hence imposing the conditions {\rm (ii)} and {\rm (iii)},
we have eight possibilities
\begin{equation}
\label{eightposs}
c_{00}=1, \qquad c_{01}=\pm {1\over\sqrt{3}},
\qquad c_{10}=\pm {1\over\sqrt{3}},
\qquad c_{11}=\pm i {1\over\sqrt{3}}.
\end{equation}
Taking the `$+$' choice above gives
$$ Tc = T\pmat{1&{1\over\sqrt{3}}\cr{1\over\sqrt{3}}&{1\over\sqrt{3}}i}
= {1\over 2}\bigl( I + {1\over\sqrt{3}}(S+\Omega)+{i\over\sqrt{3}}
S\Omega\bigr) 
= {1\over2\sqrt{3}} \pmat{ \sqrt{3}+1 & 1-i
\cr 1+i & \sqrt{3}-1 \cr}, $$
which satisfies $\trace((Tc)^4)=\trace(Tc)=1$, and so gives a
Weyl-Heisenberg SIC
$$ v={1\over\sqrt{2}\sqrt{3+\sqrt{3}}} \pmat{\sqrt{3}+1\cr 1+i}. $$
In fact all eight choices give SICs which are equivalent, as we
now explain.
\end{example}

The group generated by $S$ and $\gO$ is called the {\bf Heisenberg group},
and its normaliser in the unitary matrices is the {\bf Clifford group}.
Indeed, if a $a\in\C(d)$, then
$$ a\rho_{jk}a^{-1} = z_a(j,k)\rho_{\psi_a(j,k)}, \qquad
\forall (j,k)\in\Zd^2, $$
where $\psi_a$ is matrix multiplication by an element of $\SL_2(\Zd)$.
The Clifford group $\C(d)$ maps SIC fiducials to SIC fiducials, via the
action
$$ a\cdot (vv^*) := (av)(av)^*= a (vv^*) a^{-1}, \qquad a\in\C(d).$$
The induced action on the overlaps of the fiducial is given by
\begin{align*}
(a\cdot c)_{jk} 
&= \trace ( a(vv^*)a^{-1}S^j\gO^k)
= \inpro{a^{-1}S^j\gO^k a v,v}
= \inpro{z_{a^{-1}}(j,k)\rho_{\psi_{a^{-1}}(j,k)}v, v} \cr
&=z_{a^{-1}}(j,k)c_{\psi_{a^{-1}}(j,k)}.
\end{align*}
In \cite{BW18}, it is shown that the Clifford group is generated 
by the scalar matrices, $S$, $\gO$,
the Fourier transform $F$ and the Zauner matrix $Z$, where
$$ F_{jk}:={1\over\sqrt{d}}\go^{jk}, \qquad
Z_{jk}:= \zeta^{d-1}\mu^{j(j+d)}, \quad \mu:=e^{{2\pi\over2d}i}, \
\zeta:=e^{{2\pi\over 24}i}. $$
For these (see \cite{W17})
$$ (S^a\gO^b \cdot c)_{jk} = \go^{ak-bj} c_{jk}, \qquad
(F\cdot c)_{jk}= \go^{-jk} c_{k,-j}, \qquad
(Z\cdot c)_{jk} = \mu^{j(j+d-2k)}c_{k-j,-j}.  $$
When a (Weyl-Heisenberg) SIC fiducial $vv^*$ is known, 
there is always appears to be one which is given by an
eigenvector $v$ of $Z$ (indeed these are often searched for 
directly). Correspondingly, the overlaps satisfy
$Z\cdot c=c$, i.e., the equations
$$ \mu^{j(j+d-2k)}c_{k-j,-j} = c_{jk},  $$
which allows a further reduction of the variables $c_{jk}$.


\section{Properties of the projective Fourier transform}

Here we consider some properties of $T:\CC^{\Zd\times\Zd}\to\CC^{\Zd\times\Zd}$
given by {\rm (\ref{Tcdefn})},
i.e.,
the projective Fourier transform of $G=\Zd\times\Zd$. 
It follows from the Plancherel formula for projective representations
\cite{W18}, or (\ref{tightframerec})
that $\sqrt{d}T$ is unitary.
Indeed, (\ref{tightframerec})
can be written as $I=T\gL$, where 
$\gL:\CC^G\to\CC^G:A\mapsto(\inpro{A,\rho(g)})_{g\in G}$ satisfies
$$ \inpro{A,A}
= {1\over d}\sum_{g\in G}|\inpro{A,\rho(g)}|^2={1\over d}\inpro{\gL A,\gL A}, $$
so that ${1\over\sqrt{d}}\gL$ is unitary, and hence $\sqrt{d}T$ is unitary.

We now show that $\sqrt{d}T$ has finite order ($6d$ or $12d$), i.e.,
the projective Fourier transform for $\rho$ of (\ref{rhoalphadef})
has finite order (Theorem \ref{ShayneTfiniteorder}). 
To do this, we need a technical lemma (Lemma \ref{RsquareFlemma}), 
based on the Zauner matrix $Z$ (of order $3$), which can be factored
$$ Z=\zeta^{d-1}RF, \quad \zeta:=e^{2\pi i\over 24},  \qquad
(R)_{jk}=\mu^{j(j+d)}\gd_{jk}, \quad \mu:=e^{2\pi i\over 2d}, $$
where $F$ is the Fourier matrix, and $R$ is diagonal.
The {\it strong} form of Zauner's conjecture is that there is
a SIC fiducial which is an eigenvector of $Z$,  
for every dimension $d$. 

\begin{lemma}
\label{RsquareFlemma}
For any $d$, we have that
$$ (R^2F)^2 
= \zeta^{-6(d-1)} (RF) R^{-2} (RF)^{-1}, \qquad
\zeta=e^{2\pi i\over 24}, $$
and, in particular
$$ (R^2F)^{2d} = (-1)^{{1\over2}d(d-1)} I. $$
\end{lemma}

\begin{proof} 
Write $Z=cRF$, $c=\zeta^{d-1}$. Since $Z^3=I$ and $F^4=I$, we have
\begin{align*}
R^2F 
&= R (RF)^2 (RF)^{-1}
= R (\overline{c}Z)^2 (RF)^{-1}
= \overline{c}^2 R Z^{-1} (RF)^{-1}
= \overline{c}^3 R F^{-1}R^{-1} (RF)^{-1} \cr
& = \overline{c}^3 (R F) (F^2R^{-1}) (RF)^{-1}.
\end{align*}
Since the permutation matrix $F^2$ commutes with $R$ (or any power of $R$), 
we have
$$ (R^2F)^2 
= \overline{c}^{6} (R F) (F^2R^{-1})^2 (RF)^{-1} 
= \overline{c}^{6} (R F) R^{-2} (RF)^{-1}, $$
where $\overline{c}^{6}=\zeta^{-6(d-1)}$.
Since $R^{2d}=I$, we obtain
$$ (R^2F)^{2d} 
= \overline{c}^{6d} (R F) R^{-2d} (RF)^{-1} 
= \overline{c}^{6d} I, \qquad
\overline{c}^{6d} = \zeta^{-6d(d-1)}
= (-1)^{{1\over2}d(d-1)}, $$
which completes the proof.
\end{proof}





\begin{theorem}
\label{ShayneTfiniteorder}
The reconstruction operator 
$T$ of (\ref{Tcdefn}) has finite order, i.e.,
$$ (\sqrt{d}T)^{6d} = (-1)^{{1\over2}d(d-1)} I. $$
\end{theorem}

\begin{proof} We consider $T$ with respect to the standard 
basis $E_{jk}=e_je_k^*$ for matrices, ordered so that the 
coordinates of $c$ have the block structure
$[c]=(c_0,\ldots,c_{d-1})^T$, where $c_j$ is the $j$-th column
of the matrix $c$ (this is the order of matlab's  
{\tt reshape(c,d\verb!^!2,1)}).

The $(j,k)$-block $A_{jk}$ of the (block) matrix representation $[\sqrt{d}T]$ 
of $\sqrt{d}T$ is given by
\begin{align*}
A_{jk} v &= \hbox{$j$-th column of $\sqrt{d}T([0\ldots0,v,0 \ldots0])$ \quad
($v$ is the $k$-th column)} \cr
&={1\over\sqrt{d}}\sum_{a,b} [0\ldots0,v,0 \ldots0]_{ab} (S^a\gO^b)^*e_j 
={1\over\sqrt{d}}\sum_{a} v_a \gO^{-k} S^{-a} e_j \cr
&={1\over\sqrt{d}}\sum_{a} (\gO^{-k} P_{-1} S^{-j}e_a) v_a 
={1\over\sqrt{d}} (\gO^{-k} P_{-1} S^{-j}) v,
\end{align*}
so that
$$ A_{jk} = {1\over\sqrt{d}} \gO^{-k} P_{-1} S^{-j}
, \qquad
(A_{jk})_{ab} = {1\over\sqrt{d}} \go^{-ak}\gd_{a,j-b}. $$
The $(j,k)$-block $B_{jk}$ of $[\sqrt{d}T]^2$ is given by
\begin{align*}
(B_{jk})_{ab} 
&= \bigl(\sum_r A_{jr}A_{rk}\bigl)_{ab}
= \sum_r\sum_t (A_{jr})_{at}(A_{rk})_{tb}
= {1\over d} \sum_r\sum_t \go^{-ar}\gd_{a,j-t} \go^{-tk}\gd_{t,r-b} \cr
&= {1\over d} \go^{-a(j-a+b)}\go^{-(j-a)k}
= {1\over d} \go^{a^2 -aj +ak-jk -ab }.
\end{align*}
The $(j,k)$-block $C_{jk}$ of $[\sqrt{d}T]^3$ is given by
\begin{align*}
(C_{jk})_{ab} 
&= \bigl(\sum_r B_{jr}A_{rk}\bigl)_{ab}
= \sum_r\sum_t (B_{jr})_{at}(A_{rk})_{tb}
= \sum_r\sum_t (B_{jr})_{at}(A_{rk})_{tb} \cr
&= {1\over d\sqrt{d}} \sum_r\sum_t \go^{a^2 -aj +ar-jr -at }
 \go^{-tk}\gd_{t,r-b}
= {1\over d\sqrt{d}} \sum_r \go^{a^2 -aj +ar-jr -a(r-b) }
 \go^{-(r-b)k} \cr
&= {1\over d\sqrt{d}} \go^{a^2 -aj +ab +bk} \sum_r \go^{-r(j+k)}
={1\over\sqrt{d}} \go^{a^2 -aj +ab +bk} \gd_{j,-k}
= (R^2\gO^{-j} F \gO^k)_{ab} \gd_{j,-k},
\end{align*}
so that
$$ C_{jk} 
=\begin{cases}
0, & k\ne -j; \cr
R^2 \gO^{-j} F \gO^{-j}, & k=-j.
\end{cases}
$$
It therefore follows,
that $[\sqrt{d}T]^6$ is block diagonal, with diagonal blocks 
$$ Q_{jj} = C_{j,-j}C_{-j,j}
= (R^2 \gO^{-j} F \gO^{-j}) (R^2 \gO^j F \gO^j)
= \gO^{-j} (R^2 F)^2 \gO^j. $$
Thus $[\sqrt{d}T]^{6d}$ is block diagonal, 
and, by Lemma \ref{RsquareFlemma}, 
its diagonal blocks simplify to
$$ \gO^{-j} (R^2F)^{2d} \gO^j
=  \gO^{-j} (-1)^{{1\over2}d(d-1)}I \gO^j
= (-1)^{{1\over2}d(d-1)}I, $$
i.e., $[(\sqrt{d}T)^{6d}] = [(-1)^{{1\over2}d(d-1)}I]$.
\end{proof}

Since the projective representation (\ref{rhoalphadef}) of $\ZZ_d\times\ZZ_d$
is {\it not} an ordinary representation, 
there is no canonical presentation of the 
projective Fourier transform at $\rho$, as with the Fourier transform for
$\ZZ_d$, which gives $F$ (of order $4$), by taking $\ga=1$.
Indeed, one could take $\tilde\rho((j,k)) = b_{jk} S^j\gO^k$,
for any unit scalars $b_{jk}$, with a corresponding $\tilde\ga$-transform
(reconstruction operator)
$$ \tilde Tc:={1\over d}\sum_{j,k} c_{jk} (b_{jk} S^j\gO^k)^*
={1\over d}\sum_{j,k} c_{jk} \overline{b_{jk}} \go^{jk} S^{-j} \gO^{-k}. $$
For a general choice for $b_{jk}$, $\sqrt{d}\tilde T$ is again unitary, 
but not of finite order. During our investigation, 
we came across various
choices giving operators of finite order, in particular
\begin{equation}
\label{LenT}
Lc :={1\over d}\sum_{j,k} c_{jk} (\go^{jk} S^j\gO^k)^*
={1\over d}\sum_{j,k} c_{jk} S^{-j} \gO^{-k}.
\end{equation}
It can be shown that $L$ has the compact form 
\begin{equation}
\label{Lcompactform}
Lc = F^*(F\circ(F^* c F^*)),
\end{equation}
where $\circ$ is the Hadamard product.
From this, we obtain the following.

\begin{theorem}
The operator $L$ of (\ref{LenT}) satisfies
$$ (\sqrt{d}L)^4 c = (F^* R^{-2} F)c \qquad
\hbox{(matrix multiplication)}, $$
and hence, since $R^{2d}=I$, we have
$$ (\sqrt{d}L)^{4d}=I. $$
\end{theorem}


\begin{proof}
We first verify the compact form (\ref{Lcompactform}),
$$ (Lc)_{ab} ={1\over d}\sum_{j,k} c_{jk} (S^{-j} \gO^{-k})_{ab}
={1\over d}\sum_{j,k} c_{jk}\go^{-kb}\gd_{a,b-j}
={1\over d}\sum_{k} c_{b-a,k}\go^{-kb}, $$
\begin{align*}
\bigl( F^*( & F\circ (F^* c F^*))\bigr)_{ab}
 = \sum_{r,t,k} (F^*)_{ar}(F)_{rb} (F^*)_{rt}c_{tk}(F^*)_{kb} \cr
&={1\over d^2}\sum_{r,t,k} \go^{-ar+rb-rt}c_{tk}\go^{-kb}
={1\over d}\sum_{t,k} \gd_{t,b-a}c_{tk}\go^{-kb}
={1\over d}\sum_{k} c_{b-a,k}\go^{-kb}.
\end{align*}
Define the operation $\tilde A=P_{1-}AP_{-1}$ 
of conjugation by the 
permutation matrix $P_{-1}=F^2$ of order $2$.
This distributes over matrix multiplication,
the Hadamard product, leaving $F$ (and its powers) unchanged, so that
$Lc= F^* (F\circ (F\tilde c F))$, and
\begin{align*}
L^2c 
& = F^* (F\circ (F [F^* (F\circ (F\tilde c F))]\,\tilde{} F))
= F^* (F\circ (F [F^* (F\circ (F c F))] F)) \cr
&= F^* ([F\circ ([F\circ (F c F)] F)] F) F^*
= F^* M^2(F c F) F^*, 
\end{align*}
where
$$ Mc:= (F\circ c)F, \qquad
(Mc)_{jk} 
= \sum_t (F\circ c)_{jt} F_{tk}
= {1\over d} \sum_t \go^{jt} c_{jt} \go^{tk} 
= {1\over d} \sum_t \go^{(j+k)t} c_{jt}, $$
\begin{align*}
(M^2 c)_{jk} 
&= {1\over d} \sum_t \go^{(j+k)t} (Mc)_{jt}
= {1\over d} \sum_t \go^{(j+k)t} {1\over d}\sum_r \go^{(j+t)r} c_{jr}
={1\over d} \go^{-j(j+k)}c_{j,-(j+k)}, \cr
(M^4 c)_{jk} 
&= {1\over d^2} \go^{-j(j+k)} \go^{-j(-k)} c_{jk} 
= {1\over d^2} \mu^{-2j(j+d)}c_{jk}= ({1\over d^2}R^{-2}c)_{jk}.
\end{align*}
Thus, $M^4 c={1\over d^2}R^{-2}c$, which gives
\begin{align*}
(\sqrt{d}L)^4 c 
&= d^2 F^* M^2(F[F^* M^2(FcF) F^*]F) F^*
= d^2 F^* M^4(FcF) F^* \cr
&= F^*( R^{-2} F c F)F^*
= F^*R^{-2}F c, 
\end{align*}
and $(\sqrt{d}L)^{4d} c =  (F^*R^{-2}F)^d c 
= F^*R^{-2d}F c
= c$.
\end{proof}

\section{Equivalent equations for Heisenberg frames}

In this section, we give another condition that ensures $Tc$ has rank one,
which leads to a set of equations for $c$ which express in terms of
polynomials $p_j(z)$ which are $z$--transforms of the rows of $c$.
These polynomials $p_j(z)$ have interesting Riesz--type factorisation
properties, which we use to find a solution for $d=4$.

We use the following condition which ensures that
a matrix $A\in\CC^{d\times d}$ has rank one.

\begin{lemma}
\label{rowsymbollemma} 
$A=vv^*$ for some $v\in\Cd$  
with $v_m\ne0$ if and only if $a_{mm}>0$ and
\begin{equation}
\label{rankoneformula}
A = {1\over a_{mm}} 
\bmat{a_{0m}\cr a_{1m}\cr a_{2m}\cr\vdots\cr}
\bmat{a_{0m}\cr a_{1m}\cr a_{2m}\cr\vdots\cr}^* .
\end{equation}
\end{lemma}

\begin{proof}
First suppose that $A=vv^*$ for such a $v$.
Then
$$ \overline{v_m} v = A e_m 
= \bmat{a_{0m}\cr a_{1m}\cr a_{2m}\cr\vdots\cr}, 
\qquad |v_m|^2=(\overline{v_m}v)_m =  a_{mm}, $$
so that $a_{mm}>0$, and (\ref{rankoneformula}) holds
since $(\overline{v_m}v)(\overline{v_m}v)^*=|v_m|^2(vv^*)$.

Conversely, suppose that (\ref{rankoneformula}) holds with $a_{mm}>0$,
then clearly $A=vv^*$ for
$$ v := {1\over\sqrt{a_{mm}}}
\bmat{a_{0m}\cr a_{1m}\cr a_{2m}\cr\vdots\cr}, 
\qquad v_m=\sqrt{a_{mm}}. $$
\end{proof}

\noindent
In particular,
$Tc=vv^*$ for some $v\in\Cd$ with $v_m\ne0$ if and only if
$(Tc)_{mm}>0$ and
\begin{equation}
\label{Tcjkmatcdn}
(Tc)_{jk} = { (Tc)_{jm}\overline{(Tc)_{km}}\over (Tc)_{mm}}. 
\end{equation}

We now express (\ref{Tcjkmatcdn}) in terms of the following $z$--transform.

\begin{definition}
For $j=0,1,\ldots,d-1$,
the $j$--th {\bf symbol} of $c$ is defined to be the polynomial
$$ p_j(z) := \sum_r c_{-j,r}(\omega^jz)^{-r}. $$
\end{definition}

\noindent
This is the $z$--transform of the $j$--th row of the
matrix $(\omega^{jk} c_{-j,-k})$, since
$$ \sum_k \omega^{jk}c_{-j,-k}z^k
= \sum_r \omega^{-jr}c_{-j,r}z^{-r}
= \sum_r c_{-j,r}(\omega^jz)^{-r}. $$

We think of $p_j(z)$ as being defined only on $z^d=1$, since
each polynomial of degree $d$ is uniquely determined by its values
at the $d$--th roots of unity. Clearly, we can recover $c$
from the $d^2$ values $p_j(\omega^k)$, $j,k=0,\ldots,d-1$.
Using (\ref{Tcjk}), we calculate
$$ p_{j-k}(\omega^k) 
= \sum_r c_{k-j,r}(\omega^{j-k}\omega^k)^{-r}
= \sum_r c_{k-j,r}(\omega^j)^{-r} 
= d (Tc)_{jk}. $$

Hence (\ref{Tcjkmatcdn}) can be expressed as follows.

\begin{theorem}
\label{rowsymbolth}
$Tc=vv^*$ for  
 $v\in\Cd$  
with $v_m\ne0$ if and only if the symbols of $c$ satisfy
\begin{equation}
\label{rankformula}
p_0(\omega^m)>0, \qquad p_{j-k}(\omega^k) =
{p_{j-m}(\omega^m)\overline{p_{k-m}(\omega^m)}\over p_0(\omega^m)},
\quad j,k=0,\ldots,d-1.
\end{equation}
\end{theorem}

The symbols corresponding to a solution have interesting
Riesz--type factorisation properties, which, for simplicity,
we illustrate when $m=0$.

\begin{corollary}
\label{mequalszero}
$Tc=vv^*$ for  
 $v\in\Cd$  
with $v_0\ne0$ if and only if the symbols of $c$ satisfy
\begin{equation}
\label{rankmeqoneformula}
p_0(1)>0, \qquad p_{j-k}(\omega^k) =
{p_{j}(1)\overline{p_{k}(1)}\over p_0(1)},
\quad j,k=0,\ldots,d-1.
\end{equation}
Moreover, these have the factorisations
\begin{equation}
\label{Rfactor}
|p_j(z)|^2 = p_0(z)p_0(\omega^jz), \qquad j=0,\ldots,d-1, 
\end{equation}
and the following invariant
\begin{equation}
\label{Invar}
\prod_k p_j(\omega^k)=\prod_k p_0(\omega^k),
\qquad j=0,\ldots,d-1.
\end{equation}
\end{corollary}

\begin{proof} For (\ref{rankmeqoneformula}) take $m=0$ 
in Theorem \ref{rowsymbolth}.
Now re-index to get
$$ p_j(\omega^k)={p_{j+k}(1)\overline{p_k(1)}\over p_0(1)}, $$
and take the modulus squared of both sides
$$|p_j(\omega^k)|^2 = {|p_k(1)|^2\over p_0(1)} {|p_{j+k}(1)|^2\over p_0(1)}. $$
But, from $j=0$ in the first equation,
$$ p_0(\omega^k)={|p_k(1)|^2\over p_0(1)}, $$
so that
$$ |p_j(\omega^k)|^2 = p_0(\omega^k)p_0(\omega^{j+k})
= p_0(\omega^k)p_0(\omega^j\omega^k), $$
i.e., setting $z=\omega^k,$
$$|p_j(z)|^2=p_0(z)p_0(\omega^jz).$$

Take the product over $k$ of the re-indexed equation
$$ \prod_k p_j(\omega^k)
= {1\over (p_0(1))^d} \prod_k p_{j+k}(1)\overline{p_k(1)}
= {1\over (p_0(1))^d}\prod_k|p_k(1)|^2, $$
(since each $p_k(1)$ and its conjugate appears exactly {\it once} in
the product). Thus, by (\ref{Rfactor})
$$ \prod_k p_j(\omega^k)
= {1\over (p_0(1))^d}\prod_{k=0}^{d-1} p_0(1)p_0(\omega^k)
= \prod_{k=0}^{d-1} p_0(\omega^k). $$
\end{proof}

\medskip For completeness, we note that the Hermitian
condition of Lemma \ref{evaluelemma} can also be succinctly expressed
in terms of row symbols.

\begin{lemma}
\label{hermcdn}
$Tc$ is Hermitian if and only if
the symbols of $c$ satisfy
$$ \overline{p_j(z)} = p_{-j}(\omega^j z), \qquad j=0,\ldots,d-1. $$
\end{lemma}

\begin{proof}
From the definition, we calculate
$$ \overline{p_j(z)}  
= \overline{ \sum_r c_{-j,r} (\omega^jz)^{-r}  }
= \sum_r \overline{c_{-j,r}} (\omega^jz)^{r} 
= \sum_k \overline{c_{-j,-k}} \omega^{-jk}z^{-k}, $$
$$ p_{-j}(\omega^j z) 
= \sum_r c_{j,r}(\omega^{-j}(\omega^jz))^{-r}
= \sum_k c_{jk}z^{-k}, $$
and so, by equating the coefficients of $z^{-k}$,
the Hermitian condition $c_{jk}= \omega^{-jk}\overline{c_{-j,-k}}$ is
equivalent to equality of the above symbols.
\end{proof}

\section{The Special Case of $d=3$}

This case already has some interesting geometric features.

Solving the basic equations for the $c_{jk},$ is also
geometrically interesting.

\begin{proposition}
\label{dthreeprop} For $d=3$,
$c$ generates a Heisenberg frame with $v_0\ne0$ if and only if
\item{\rm (a)}
$\displaystyle{p_0(z)= 1+\overline{c_{01}}z+c_{01}z^2}$
and $\displaystyle{p_2(z)=\overline{p_1(\omega^2z)}}$
\quad (Hermitian conditions)
\item{\rm(b)} $\displaystyle{|p_1(z)|^2=p_0(z)p_0(\omega z)}$
\quad (Riesz factorization)
\item{\rm(c)} $\displaystyle{\prod_{k=0}^2p_1(\omega^k)=\prod_{k=0}^2p_0(\omega^k)}$ 
\quad (invariant condition)
\item{\rm(d)} $\displaystyle{|c_{01}|=|c_{1k}|={1\over{2}},\quad k=0,1,2.}$
\end{proposition}

\begin{proof}
By Lemma \ref{hermcdn}, the conditions for $Tc$ to be Hermitian are
$$ \overline{p_0(z)}=p_0(z), \qquad 
\overline{p_1(z)}=p_2(\omega z), \qquad 
\overline{p_2(z)}=p_1(\omega^2 z). $$
Since $p_0(z)=c_{00}+c_{01}z^2+c_{02}z$,
the first equation is satisfied provided
$$ \overline{c_{00}}+\overline{c_{01}}z+\overline{c_{02}}z^2
= c_{00} +c_{02}z +c_{01}z^2 \Iff c_{00}\in\RR, 
\quad c_{02}=\overline{c_{01}}. $$
The second and third are equivalent, since substituting
$\omega z$ for $z$ in the third gives
$$ \overline{p_2(\omega z)} = p_1(\omega^2(\omega z))=p_1(z). $$
Hence (a) is equivalent to $Tc$ being Hermitian with $c_{00}=1$,
and implies $c_{02}=\overline{c_{01}}$.

By Corollary \ref{mequalszero},
(a),(b),(c), (d) hold for a Heisenberg frame with $v_0\ne0$.
For the converse, suppose that (a),(b),(c), (d) hold.
Then by (a), $Tc$ is Hermitian with $c_{00}=1$, and $c_{02}=\overline{c_{01}}$,
so that (d) gives $|c_{jk}|={1\over2}$, $(j,k)\ne(0,0)$.
Hence Lemma \ref{evaluelemma}, gives
$$ P_{Tc}(\lambda)=\lambda^3-\lambda^2+0\lambda+a_0, \qquad
a_0=-\det(Tc). $$
In view of Theorem \ref{cvcharthm}, with condition (iv)$'$, we need only
show that $\det(Tc)=0$.

Since $d(Tc)_{jk}=p_{j-k}(\omega^k)$, condition (a) gives
$$ 3(Tc) 
= \bmat{p_0(1)& p_2(\omega) & p_1(\omega^2)\cr
p_1(1)& p_0(\omega) & p_2(\omega^2)\cr
p_2(1)& p_1(\omega)& p_0(\omega^2)}
= \bmat{p_0(1)& \overline{p_1(1)} & p_1(\omega^2)\cr
p_1(1)& p_0(\omega) & \overline{p_1(\omega)}\cr
\overline{p_1(\omega^2)}& p_1(\omega)& p_0(\omega^2)}.
$$
Since $\overline{p_0(z)}=p_0(z)$, the invariant condition gives
$\prod_k p_1(\omega^k) 
= \prod_k \overline{p_1(\omega^k)} 
= \prod_k \overline{p_0(\omega^k)} $, and we calculate
and so
\begin{align*}
3\det(Tc) 
&= p_0(1) \{p_0(\omega) p_0(\omega^2)-|p_1(\omega)|^2\}
-\overline{p_1(1)}\{p_1(1)p_0(\omega^2)-\overline{p_1(\omega^2)}
\overline{p_1(\omega)}\} \cr
& \qquad +p_1(\omega^2)\{p_1(1)p_1(\omega)-\overline{p_1(\omega^2)}p_0(\omega)\}
\cr 
&= 3 \prod_k p_0(\omega^k)
- p_0(1) |p_1(\omega)|^2
- p_0(\omega^2) |p_1(1)|^2
- p_0(\omega) |p_1(\omega^2)|^2.
\end{align*}
Applying the Riesz factorisation to the last three terms we then obtain
$$ 3\det(Tc) = 3 \prod_k p_0(\omega^k)
- p_0(1) p_0(\omega)p_0(\omega^2)
- p_0(\omega^2) p_0(1)p_0(\omega)
- p_0(\omega) p_0(\omega^2)p_0(1)
= 0. $$
(need to check $v_0\ne0$!)
\end{proof}

We now use Proposition \ref{dthreeprop} to find the solutions for $d=3$.
First we consider the Riesz--type factorisation
$|p_1(z)|^2=p_0(z)p_0(\omega z)$.
Note that $1+\omega+\omega^2=0$, and the variable $z$ of our symbols
satisfies $z^3=1$, $\overline{z}=z^2$. Hence multiplying out gives
\begin{align*}
|p_1(z)|^2 & = p_1(z)\overline{p_1(z)}
= (c_{20}+c_{21}\omega^2z^2+c_{22}\omega z)
  (\overline{c_{20}}+\overline{c_{21}}\omega z+\overline{c_{22}}\omega^2 z^2)\cr
& = 
(\hbox{$\sum_k$} |c_{2k}|^2)
+ (c_{20}\overline{c_{21}}+c_{21}\overline{c_{22}}+c_{22}\overline{c_{20}})
\omega z
+ (c_{20}\overline{c_{22}}+c_{21}\overline{c_{20}}+c_{22}\overline{c_{21}})
\omega^2 z^2,
\end{align*}
and
\begin{align*}
p_0(z)p_0(\omega z) 
& = (1+\overline{c_{01}}z+c_{01}z^2)
(1+\overline{c_{01}}\omega z+c_{01}\omega^2z^2) \cr
&= 1-|c_{01}|^2 + \omega^2(c_{01}^2-\overline{c_{01}})z
+\omega(\overline{c_{01}}^2-c_{01})z^2.
\end{align*}
Hence, equating the coefficients of $1,z,z^2$, gives
\begin{align*}
|c_{20}|^2+|c_{21}|^2+|c_{22}|^2
& = 1-|c_{01}|^2, \cr
c_{20}\overline{c_{21}}+c_{21}\overline{c_{22}}+c_{22}\overline{c_{20}}
& = (c_{01}^2-\overline{c_{01}}) \omega, \cr
c_{20}\overline{c_{22}}+c_{21}\overline{c_{20}}+c_{22}\overline{c_{21}}
& = (\overline{c_{01}}^2-c_{01})\omega^2.
\end{align*}
Since $|c_{01}|^2=|c_{1k}|^2={1\over 4}$, the first equation is
automatically satisfied. Further, the second and third are conjugates of
each other, and so we have only one equation (for the Riesz--type
factorisations)
$$ c_{20}\overline{c_{22}}+c_{21}\overline{c_{20}}+c_{22}\overline{c_{21}}
= (\overline{c_{01}}^2-c_{01})\omega^2. $$
Setting $z_{jk}:=c_{jk}/|c_{jk}|=2c_{jk}$, this becomes
$$ {1\over 4} ( z_{20}\overline{z_{22}}+z_{21}\overline{z_{20}}
+z_{22}\overline{z_{21}})
= ({1\over4} \overline{z_{01}}^2-{1\over 2} z_{01})\omega^2. $$
Since $\overline{z_{01}}^2=z_{01}^{-2}$, this can be rewritten as
$$ z_{20}\overline{z_{22}}+z_{21}\overline{z_{20}} +z_{22}\overline{z_{21}}
= \bigl( {1\over z_{01}^2} -2 z_{01}\bigr)\omega^2
= {1-2 z_{01}^3 \over z_{01}^2}\omega^2
= {1-2 (z_{01}/\omega)^3 \over (z_{01}/\omega)^2}.  $$
Now we set $z:=z_{01}/\omega$, so that our equation becomes
\begin{equation}\label{HypoEqn}
z_{20}\overline{z_{22}}+z_{21}\overline{z_{20}} +z_{22}\overline{z_{21}}
= {1-2z^3 \over z^2}.
\end{equation}

We proceed to analyze both sides of this equation.

\begin{lemma}
\label{hypoboundary}
The curve
$$\gth \mapsto {1-2z^3\over z^2}, \quad z=-e^{i\gth} $$
is a $3$--cusped hypocycloid (or deltoid).
\end{lemma}

\begin{proof}
Recall that the standard parametric equations
for a hypocycloid with radii
$a$ and $b$ with  $a>b>0$ are (see,  e.g.  \cite{HypoWiki})
\begin{align*}
x(\gth)&=(a-b)\cos(\gth)+b\cos({a-b\over b}\gth),\cr
y(\gth)&=(a-b)\sin(\gth)-b\sin({a-b\over b}\gth),
\end{align*}
and if $n=a/b$ is an integer, it is $n$--cusped.
Now
$$w:={1-2z^3\over z^2}=-2z+z^{-2} = 2e^{i\gth} + e^{-2i\gth}, $$
which has Cartesian coordinates
$$ \Re(w) = 2 \cos(\gth)+\cos(2\gth), \qquad
   \Im(w) = 2 \sin(\gth)-\sin(2\gth), $$
and so $w(\gth)$ is a $3$--cusped hypocycloid with radii $a=3$ and $b=1$.
\end{proof}

\begin{figure}[ht]
\centering
  \includegraphics[width=6cm, height=6cm]{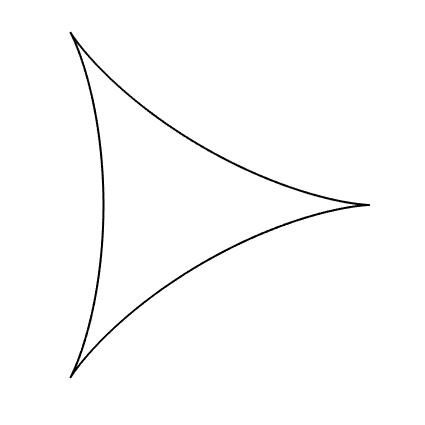} 
  \caption{The Hypocycloid}
  \label{fig1}
\end{figure}

\medskip For the left side, note that the product of the
three terms
$$ z_{20}\overline{z_{22}} \cdot z_{21}\overline{z_{20}}
\cdot z_{22}\overline{z_{21}}=1. $$

\begin{lemma}
\label{hyposet}
The set of complex numbers
$$\{z_1+z_2+z_3\,:\, z_1z_2z_3 =1,\, |z_j|=1\}$$
is the interior and boundary of the $3$--cusped hypocycloid
given by the right side, i.e.,
$$ \gth \mapsto {1-2z^3\over z^2},\quad z=-e^{i\gth}. $$
In particular, points on the boundary have the form
$$ z_1=z_3=e^{-i{\phi\over2}}, \quad z_2=e^{i\phi} \qquad
\hbox{($\gth=-{\phi\over2}$)},$$
$$ z_1=z_2=e^{-i{\phi\over2}}, \quad z_3=e^{i\phi} 
,$$
or
$$ z_2=z_3=e^{-i{\phi\over2}}, \quad z_1=e^{i\phi} 
.$$
\end{lemma}

\begin{proof}
Since $z_1z_2z_3=1$, we can write a point $w$ in the set as
$$ w = z_1 + z_2 + z_3, \qquad  z_1:=e^{i t}, \quad z_2:=e^{i\phi}, \quad
z_3:=e^{-i(t+\phi)}. $$
Now fix $\phi$, and let $t$ vary.
Let $A$ and $B$ be the points on
the hypocycloid for $\gth=-{\phi\over2}$ and $\gth=\pi-{\phi\over2}$, i.e.,
$$ A = 2 e^{-i{\phi\over2}} + e^{i\phi}, \qquad
B = -2 e^{-i{\phi\over2}} + e^{i\phi}. $$
We claim (cf. Figure 2) that as $t$ varies $w$ traces out the line segment connecting
$A$ and $B$, precisely
$$ w 
= e^{it} + e^{i\phi} + e^{-i(t+\phi)}
= \gl A +(1-\gl) B, \qquad \gl = {\cos(t+{\phi\over2})+1\over 2}
\in [0,1], 
$$
which we verify by multiplying out
\begin{align*}
\gl A+(1-\gl)B & = \gl(A-B)+B
=  { e^{i(t+{\phi\over2})} + e^{-i(t+{\phi\over2})} +2 \over 4 }
4 e^{-i{\phi\over2}} -2 e^{-i{\phi\over2}} + e^{i\phi} \cr
&=  ( e^{it} + e^{-i(t+\phi)} +2 e^{-i{\phi\over2}} )
-2 e^{-i{\phi\over2}} + e^{i\phi} 
=   e^{it} + e^{i\phi} + e^{-i(t+\phi)}.
\end{align*}
Further we note that this line segment is tangent to
the point where $\theta=\phi$, i.e.,
$$ C = 2 e^{i\phi}+e^{-2i\phi}. $$
Indeed the tangent to the hypocyloid at this point is
$$ {d\over d\gth} \bigl(2e^{i\gth}+e^{-2i\gth}\bigr)\bigl|_{\gth=\phi}
= 2e^{i\phi}-2e^{-2i\phi}
= i\sin(\hbox{${3\over2}$}\phi)(A-B), $$
which is collinear with the line segment, except when
$\phi=0,\pm{2\over3}\pi$, the three cusps of the hypocycloid.

At the cusp corresponding to $\phi=0$, $A=3$ and
$B=-1$, and thus the line segment connecting $A$ and $B$ is
also ``tangent'' at that cusp. The other cusps are handled similarly.
\end{proof}

\begin{figure}[h!]
\centering
  \includegraphics[width=6cm, height=6cm]{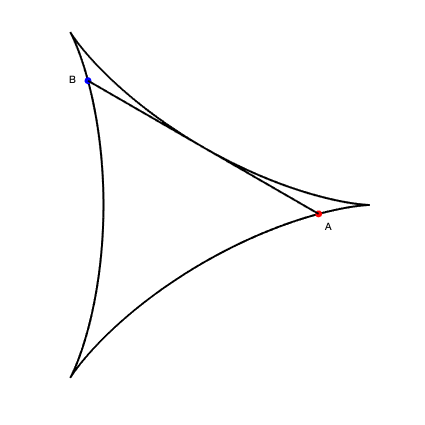} 
  \caption{Points A and B on the Hypocycloid}
  \label{fig2}
\end{figure}

\medskip
Thus, from equation (\ref{HypoEqn}), it follows that 
$z_{20}\overline{z_{22}} ,$ $z_{21}\overline{z_{20}} $ and
$z_{22}\overline{z_{21}}$ are {\it boundary} points of the hypocyloid.  Solutions may be obtained as follows.  Pick one of the boundary solutions, e.g., where $z_1=z_3,$   so that
$$ z_{20}\overline{z_{22}} = e^{-i{\phi\over2}}, \qquad
 z_{21}\overline{z_{20}} = e^{i\phi}, \qquad
 z_{22}\overline{z_{21}} = e^{-i{\phi\over2}}, \qquad
 z = -e^{-i{\phi\over 2}}. $$
 These can be solved using one of them as a free parameter, i.e.,
 $$ z_{22}=e^{i{\phi\over2}}z_{20} \,\,\hbox{and}\,\, z_{21}=e^{i\phi}z_{20},\quad |z_{20}|=1. $$
 In this way we arrive at a continuum of parameterized solutions for the overlaps of a SIC in dimension $d=3.$
 
 \section{Closing Comment}
 It has sometimes been remarked that the overlaps are zeros of a self-reciprocal polynomial ($z^np(1/z)=p(z)$) with integer coefficients. The fact that the coefficients are integers is notable and perhaps important. However being self-reciprocal is {\it not}. Indeed if they are roots of a polynomial $p(z)$ of degree $n,$ then they are also automatically roots of $q(z):=p(z) \times z^np(1/z)$ and this latter polynomial is self-reciprocal.

\section{Acknowledgement}

We would like to thank 
Marcus Appleby
for many helpful discussions related to SICs.


\bibliographystyle{alpha}
\bibliography{references}
\nocite{*}



\vfil
\end{document}